\documentclass[a4paper,11pt]{amsart}
\usepackage[a4paper,margin=25mm]{geometry}

\usepackage{amsmath,amsthm}
\usepackage[T1]{fontenc}
\usepackage{enumerate}
\usepackage{hyperref}
\usepackage{lmodern}
\usepackage[cp1250]{inputenc}
\usepackage{graphicx}

\usepackage[ruled]{algorithm}
\usepackage[noend]{algpseudocode}
\newcommand{\algofont}[1]{\textnormal{\selectfont\sffamily#1}}
\newcommand{\NPclass}{\textnormal{\selectfont\sffamily NP}}
\newcommand{\algpairing}{{\algofont{Pairing}}}
\newcommand{\algreplace}{{\algofont{PairReplacement}}}
\newcommand{\algmain}{{\algofont{TtoG}}}
\newcommand{\constfont}[1]{\textnormal{\textsl{#1}}}
\newcommand{\constsfont}[1]{\textnormal{\textbf{\textsl{#1}}}}
\newcommand{\twodots}{\mathinner{\ldotp\ldotp}}

\newcommand{\Pref}[1]{(\hyperref[P#1]{P#1})}
\newcommand{\pocz}{{\constfont{start}}}
\newcommand{\pos}{{\constfont{newpos}}}
\newcommand{\myFalse}{{\constfont{false}}}
\newcommand{\myTrue}{{\constfont{true}}}
\newcommand{\Not}{{\constsfont{not}}}
\newcommand{\kon}{{\constfont{end}}}
\newcommand{\pair}{{\constfont{pair}}}
\newcommand{\lewy}{{\constfont{first}}}
\newcommand{\prawy}{{\constfont{second}}}
\newcommand{\none}{{\constfont{none}}}

\usepackage{color}
\definecolor{myYellow}{rgb}{0.9,0.9,0}

\providecommand{\Ocomp}{{\mathcal O}}

\newtheorem{theorem}{Theorem}
\newtheorem{lemma}{Lemma}

\theoremstyle{definition}

\theoremstyle{remark}

\usepackage[ruled]{algorithm}
\usepackage[noend]{algpseudocode}

\newcommand{\RePair}{\algofont{RePair}}
\newcommand{\algradix}{\algofont{RadixSort}}

\begin{document}

\title{A \emph{really} simple approximation of smallest grammar}

\author[A.\ Je\.z]
{Artur Je\.z}
\thanks{Supported by NCN grant number 2011/01/D/ST6/07164, 2011--2014.}
\address{
Max Planck Institute f\"ur Informatik,\\
Campus E1 4,  DE-66123 Saarbr\"ucken, Germany\\
\and
Institute of Computer Science, University of Wroc{\l}aw \\
ul.\ Joliot-Curie~15, 50-383 Wroc{\l}aw, Poland\\
\texttt{aje@cs.uni.wroc.pl}}

\begin{abstract}
In this paper we present a \emph{really} simple linear-time algorithm
constructing a context-free grammar of size $\Ocomp(g \log (N/g))$
for the input string, where $N$ is the size of the input string
and $g$ the size of the optimal grammar generating this string.
The algorithm works for arbitrary size alphabets,
but the running time is linear assuming that the alphabet $\Sigma$ of the input string
can be identified with numbers from $\{1,\ldots , N^c \}$ for some constant $c$.
Algorithms with such an approximation guarantee and running time are known,
however all of them were non-trivial and their analyses were involved.
The here presented algorithm computes the LZ77 factorisation and transforms it in phases to a grammar.
In each phase it maintains an LZ77-like factorisation of the word with at most $\ell$ factors
as well as additional $\Ocomp(\ell)$ letters, where $\ell$ was the size of the original LZ77 factorisation.
In one phase in a greedy way (by a left-to-right sweep and a help of the factorisation)
we choose a set of pairs of consecutive letters to be replaced with new symbols,
i.e.\ nonterminals of the constructed grammar.
We choose at least 2/3 of the letters in the word and there are $\Ocomp(\ell)$ many different pairs among them.
Hence there are $\Ocomp(\log N)$ phases, each of them introduces $\Ocomp(\ell)$ nonterminals to a grammar.
A more precise analysis yields a bound $\Ocomp(\ell \log(N/\ell))$.
As $\ell \leq g$, this yields the desired bound $\Ocomp(g \log(N/g))$.
\end{abstract}
	\maketitle

\section{Introduction}
	\label{sec:intro}

\subsection*{Grammar based compression}
In the grammar-based compression text is represented by a~context-free grammar (CFG) generating exactly one string.
Such an approach was first considered by Rubin~\cite{Rubin76},
though he did not mention CFGs explicitly.
In general, the idea behind this approach is that a CFG can compactly represent the structure of the text,
even if this structure is not apparent.
Furthermore, the natural hierarchical definition of the context-free grammars makes such a~representation suitable for algorithms,
in which case the string operations can be performed on the compressed representation,
without the need of the explicit decompression~\cite{GawryLZ,FCPM,SLPpierwszePM,PlandowskiSLPequivalence,RytterSWAT,SLPaprox2}.

While grammar-based compression was introduced with practical
purposes in mind and the paradigm was used in several implementations~\cite{RePair,KiefferY96,Sequitur},
it also turned out to be very useful in more theoretical considerations.
Intuitively, in many cases large data have relatively simple inductive definition,
which results in a grammar representation of small size. 
On the other hand, it was already mentioned that the hierarchical structure of the
CFGs allows operations directly on the compressed representation.
A recent survey by Lohrey\cite{Lohreysurvey} gives a comprehensive description
of several areas of theoretical computer science in which grammar-based compression was successfully applied.

The main drawback of the grammar-based compression is that producing the smallest CFG for a text is \emph{intractable}:
given a string $w$ and number $k$ it is {\selectfont\sffamily NP}-hard to decide whether there exist a CFG of size $k$ that generates $w$~\cite{SLPapproxNPhard}.
Furthermore,
the size of the smallest grammar for the input string cannot be approximated within some small constant factor~\cite{SLPaprox2}.

\subsection*{Previous approximation algorithms}
The first two algorithms with an approximation ratio $\Ocomp(\log(N/g))$ were developed simultaneously
by Rytter~\cite{SLPaprox} and Charikar et~al.~\cite{SLPaprox2}.
They followed a similar approach, we first present Rytter's approach as it is a bit easier to explain.

Rytter's algorithm~\cite{SLPaprox} applies the LZ77 compression to the input string and then transforms the
obtained LZ77 representation to an $\Ocomp(\ell \log(N/\ell))$ size grammar,
where $\ell$ is the size of the LZ77 representation.
It is easy to show that $\ell \leq g$ and as $f(x) = x \log(N/x)$ is increasing,
the bound $\Ocomp(g\log(N/g))$ on the size of the grammar follows
(and so a bound $\Ocomp(\log(N/g))$ on approximation ratio).
The crucial part of the construction is the requirement that the derivation tree of the intermediate constructed grammar
satisfies the AVL condition.
While enforcing this requirement is in fact easier than in the case of the AVL search trees (as the internal nodes do not store any data),
it remains involved and non-trivial.
Note that the final grammar for the input text is also AVL-balanced, which makes is suitable for later processing.

Charikar et~al.~\cite{SLPaprox2} followed more or less the same path,
with a different condition imposed on the grammar: it is required that its derivation tree is length-balanced,
i.e.\ for a rule $X \to YZ$ the lengths of words generated by $Y$ and $Z$ are within a certain 
multiplicative constant factor from each other.
For such trees efficient implementation of merging, splitting etc.\ operations were given (i.e.\ constructed from scratch)
by the authors and so the same running time as in the case of the AVL grammars was obtained.
Since all the operations are defined from scratch, the obtained algorithm is also quite involved and the analysis is even more non-trivial.

Sakamoto~\cite{SLPaproxSakamoto} proposed a different algorithm, based on \RePair~\cite{RePair},
which is one of the practically implemented and used algorithms for grammar-based compression.
His algorithm iteratively replaces pairs of different letters and maximal repetitions of letters
($a^\ell$ is a \emph{maximal repetition} if that cannot be extended by $a$ to either side).
A special pairing of the letters was devised, so that it is `synchronising':
if $u$ has $2$ disjoint occurrences in $w$, then those two occurrences can be represented
as $u_1u'u_2$, where $u_1,u_2 = \Ocomp(1)$, such that both occurrences of $u'$ in $w$
are paired and compressed in the same way.
The analysis was based on considering the LZ77 representation of the text and proving that due to
`synchronisation' the factors of LZ77 are compressed very similarly as the text to which they refer.
Constructing such a pairing is involved and the analysis non-trivial.

Recently, the author proposed another algorithm~\cite{grammar}.
Similarly to the Sakamoto's algorithm it iteratively applied two local replacement rules
(replacing pairs of different letters and replacing maximal repetitions of letters).
Though the choice of pairs to be replaced was simpler, still the construction was involved.
The main feature of the algorithm was its analysis based on the recompression technique,
which allowed avoiding the connection of SLPs and LZ77 compression.
This made it possible to generalise this approach also to grammars generating trees~\cite{treegrammar}.
On the downside, the analysis is quite complex.

\subsection*{Contribution of this paper}
We present a very simple algorithm together with a straightforward and natural analysis.
It chooses the pairs to be replaced in the word during a left-to-right sweep and additionally
using the information given by a LZ77 factorisation.
We require that any pair that is chosen to be replaced is either inside a factor of length at least $2$
or consists of two factors of length $1$ and that the factor is paired in the same way as its definition.
To this end we modify the LZ77 factorisation during the sweep.
After the choice, the pairs are replaced and the new word inherits the factorisation from the original word.
This procedure is repeated until a trivial word is obtained.
To see that this is indeed a grammar construction, when the pair $ab$ is replaced by $c$ we create a rule $c \to ab$.

\subsubsection*{Note on computational model}
The presented algorithm runs in linear time, assuming that we can compute the LZ77 factorisation in linear time.
This can be done assuming that the letters of the input words can be sorted in linear time,
which follows from a standard assumption that $\Sigma$ can be identified with a continues
subset of natural numbers of size $\Ocomp(N^c)$ for some constant $c$ and the \algradix{} can be performed on it.
Note that such an assumption is needed for all currently known linear-time algorithms
that attain the $\Ocomp(\log(N/g))$ approximation guarantee.

\section{Notions}
By $N$ we denote the size of the input word.

\subsection*{LZ77 factorisation}
An LZ77 factorisation (called simply factorisation in the rest of the paper)
of a word $w$ is a representation $w = f_1f_2\cdots f_\ell$, where each $f_i$
is either a single letter (called \emph{free letter} in the following)
or $f_i = w[j \twodots j + |f_i|-1]$ for some $j \leq |f_1\cdots f_{i-1}|$,
in such a case $f$ is called a \emph{factor} and $w[j \twodots j + |f_i|-1]$ is called the \emph{definition} of this factor.
We do not assume that a factor has more than one letter though when we find such a factor we demote it to a free letter.
The \emph{size} of the LZ77 factorisation $f_1f_2\cdots f_\ell$ is $\ell$.
There are several simple and efficient linear-time algorithms for computing the smallest LZ77 factorisation of a word~%
\cite{CrochemoreLZ77,PuglisiLZ77,CrochemoreLZ772,BannaiLZ77,KarkkainenLZ77,OhlebuschLZ77}
and all of them rely on linear-time algorithm for computing the suffix array~\cite{suffixarrays}.

\subsection*{SLP}
\emph{Straight Line Programme} (\emph{SLP}) is a CFG in the Chomsky normal form that generates a unique string.
Without loss of generality we assume that nonterminals of an SLP are $X_1,\ldots,X_g$,
each rule is either of the form $X_i \to a$ or $X_i \to X_jX_k$, where $j,k < i$.
The \emph{size} of the SLP is the number of its nonterminals (here: $g$).

The problem of finding smallest SLP generating the input word $w$ is \NPclass-hard~\cite{SLPapproxNPhard} and
the size of the smallest grammar for the input word cannot be approximated within some small constant factor~\cite{SLPaprox2}.
On the other hand, several algorithms with approximation ratio $\Ocomp(1 + \log(N/g))$, where $g$ is the size 
of the smallest grammar generating $w$, are known~\cite{SLPaprox2,SLPaprox,SLPaproxSakamoto,grammar}.
Most of those constructions use the inequality $\ell \leq g$,
where $\ell$ ($g$) is the size of the smallest LZ77 factorisation (grammar, respectively) for $w$~\cite{SLPaprox}.

\section{Intuition}
\subsection*{Pairing}
Relaxing the Chomsky normal form, let us identify each nonterminal generating a~single letter with this letter.
Suppose that we already have an SLP for $w$.
Consider the derivation tree for $w$ and the nodes that have only leaves as children
(they correspond to nonterminals that have only letters on the right-hand side).
Such nodes define a \emph{pairing} on $w$, in which each letter is paired with one of the neighbouring letters
(such pairing is of course a symmetric relation).
Construction of the grammar can be naturally identified with iterative pairing: for a word $w_i$ we find a pairing,
replace pairs of letters with `fresh' letters
(different occurrences of a pair $ab$ can be replaced with the same letter though this is not essential), obtaining $w_{i+1}$
and continue the process until a word $w_{i'}$ has only one letter.
The fresh letters from all pairings are the nonterminals of the constructed SLP
and its size is twice the number of different introduced letters.
Our algorithm will find one such pairing using the LZ77 factorisation of a word.

\begin{figure}
	\centering
		\includegraphics[scale=1.5]{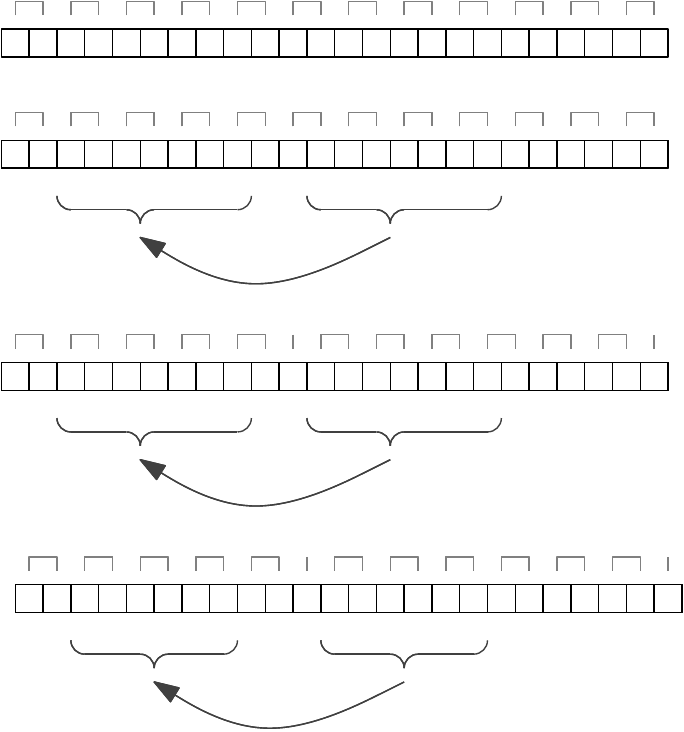}
	\caption{
	The pairings are presented over the words, the LZ77 factors are depicted below the words.
	On the top, the naive pairing is presented.
	On the second picture, the pairing is compared with the LZ77 factor as well as its definition; the factor and its definition are paired differently.
	On the third picture, we move the pairing so that it is consistent on the factor and its definition. This creates unpaired letters.
	On the bottom, we shorten the factor on the right, so that it ends with a whole pair.}
	\label{fig:pairing}
\end{figure}

\subsection*{Creating a pairing}
Suppose that we are given a word $w$ and know its factorisation.
We try the naive pairing: the first letter is paired with second, third with fourth and so on, see Fig.~\ref{fig:pairing}.
If we now replace all pairs with new letters, we get a word that is $2$ times shorter so $\log N$ such iterations
give an SLP for $w$.
However, in the worst case there are $|w|/2$ different pairs already in the first pairing and so we cannot give
any better bound on the grammar size than $\Ocomp(N)$.

A better estimation uses the LZ77 factorisation.
Let $w = f_1f_2\cdots f_\ell$ and consider a factor $f_i$. It is equal to $w[j\twodots j+|f_i|-1]$
and so all pairs occurring in $f_i$ already occur in $w[j\twodots j+|f_i|-1]$ unless the parity is different,
i.e.\ $j$ and $|f_1\cdots f_{i-1}|+1$ are of different parity, see Fig.~\ref{fig:pairing}.
We want to fix this: it seems a bad idea to change the pairing in $w[j\twodots j+|f_i|-1]$,
so we change it in $f_i$: it is enough to shift the pairing by one letter,
i.e.\ leave the first letter of $f_i$ unpaired and pair the rest as in $w[j+1\twodots j+|f_i|-1]$.
Note that the last letter in the factor definition may be paired with the letter to the right,
which may be impossible inside $f_i$, see Fig.~\ref{fig:pairing}.
As a last observation note that since we alter each  $f_i$,
instead of creating a pairing at the beginning and modifying it we can create the pairing while scanning the word from left to right.

There is one issue: after the pairing we want to replace the pairs with fresh letters.
This may make some of the the factor definitions improper:
when $f_i$ is defined as $w[j \twodots j + |f_i|-1]$ it might be that $w[j]$ is paired with letter to the left.
To avoid this situation, we replace the factor $f_i$ with $w[j]f_i'$ and change its definition to $w[j+1 \twodots j + |f_i|-1]$. 
Similar operation may be needed at the end of the factor, see Fig.~\ref{fig:pairing}.
This increases the size of the LZ77 factorisation, but the number of factors stays the same
(i.e.\ only the number of free letters increases).
Additionally, we pair two neighbouring free letters, whenever this is possible.

\subsection*{Using a pairing}
When the appropriate pairing is created, we replace each pair with a new letter. If the pair is within a factor,
we replace it with the same symbol as the corresponding pair in the definition of the factor.
In this way only pairs that are formed from free letters may contribute a fresh letter.
As a result we obtain a new word together with a factorisation in which there are $\ell$ factors.

\subsection*{Analysis}
The analysis is based on the observation that a factor $f_i$ is shortened by a constant fraction,
so it takes part in $\log |f_i|$ phases and in each of them it introduces $\Ocomp(1)$ free letters.
Hence the total number of free letters introduced to the word is
$\Ocomp(\sum_{j = 1}^\ell \log |f_i|) = \Ocomp(\ell \log (N/\ell))$ (which is shown in details later on).
As creation of a rule decreases the number of free letters in the instance by at least $1$,
we obtain that this is also an upper bound on the size of the grammar.

\section{The algorithm}
\subsection*{Stored data}
The word is represented as a table of letters.
The table $\pocz$ stores the information about the beginnings of factors:
$\pocz[i] = j$ means that $i$ is the first letter of a factor and $j$ is the first letter of its definition;
otherwise $\pocz[i] = \myFalse$.
Similarly $\kon$ stores the information about the ends of factors:
$\kon[i]$ is a bit flag, i.e.\ has value $\myTrue/\myFalse$, that tells whether $w[i]$ is the last letter of a factor.

When we replace the pairs with new letters, we reuse the same tables, overwriting from left to the right.
Additionally, a table $\pos$ stores the corresponding positions: $\pos[i] = j$ means that 
\begin{itemize}
	\item the letter on position $i$ was unpaired and $j$ is the position of the corresponding letter in the new word
	\textit{or}
	\item the letter on position $i$ was paired with a letter to the right and the corresponding letter in the new word is on position $j$.
\end{itemize}
Lastly, $\pos[i]$ is not defined when position $i$ was the second element in the pair.

It is easy to see that the algorithm can be converted to lists and pointers instead of tables.
(Though the \algradix{} used in the LZ77 construction needs tables).

\subsection*{Technical assumption}
Our algorithm makes a technical assumptions:
a factor $f_i$ (of length at least $2$) starting at position $j$ cannot have $\pocz[j] = j-1$,
i.e.\ its definition is at least two positions to the left.
This is verified and repaired while sweeping through $w$:
if $w[j \twodots j + |f_i|-1] = w[j-1 \twodots j + |f_i|-2]$ then $f_i = a^{|f_i|}$ for some letter $a$.
In such a case we split $f_i$: we make $w[j]$ a free letter and set $\pocz[j+1] = j-1$
(note that the latter essentially requires that indeed $f_i$ is a repetition of one letter).

\subsection*{Pairing}
We are going to devise a pairing with the following
(a bit technical) properties:
\begin{enumerate}[(P1)]
	\item \label{P1}there are no two consecutive letters that are both unpaired;
	\item \label{P2}the first two letters (last two letters) of any factor $f$ are paired with each other;
	\item \label{P3}if $f = w[i \twodots i + |f| - 1]$ has a definition $w[\pocz[i]\twodots\pocz[i]+|f|-1]$
	then letters in $f$ and in $w[\pocz[i]\twodots\pocz[i]+|f|-1]$ are paired in the same way.
\end{enumerate}
The pairing is found incrementally by a left-to-right scan through $w$:
we read $w$ and when we are at letter $i$ we make sure that the word $w[1\twodots i]$ satisfies \Pref{1}--\Pref{3}.
To this end we not only devise the pairing but also modify the factorisation a bit
(by replacing a factor $f$ with $af$ or by $fb$, where $a$ is the first and $b$ the last letter of $f$).
If during the sweep some $f$ is shortened so that $|f| = 1$ then we demote it to a free letter.

The pairing is recorded in table: $\pair[i]$ can be set to $\lewy$, $\prawy$ or $\none$,
meaning that $w[i]$ is the first, second  in the pair or it is unpaired, respectively.

\begin{algorithm}[h]
	\caption{\algpairing}
	\label{alg: pairing}
	\begin{algorithmic}[1]
		\State $\pair[1] \gets \none$ 
		\State $i \gets 2$
		\While{$i \leq |w|$}
			\If{$\pocz[i]$} \Comment{$w[i]$ is the first element of a factor}
				\If{$\kon[i]$} \Comment{This is one-letter factor} \label{ensure non-trivial}
					\State $\pocz[i] \gets \kon[i] \gets \myFalse$ \Comment{Turn it into a free letter}
				\ElsIf{$\pocz[i] = i-1$} \Comment{The factor is $a^k$, its definition begins one position to the left} \label{blok}
					\State $\pocz[i+1] \gets i-1$			\Comment{Move the definition of the factor}
					\State $\pocz[i] \gets \myFalse$  \Comment{Make $w[i]$ a free letter}
					\label{blok pocz}
				\ElsIf{$\pair[\pocz[i]] \neq \lewy$} \Comment{The pairing of the definition of factor is bad} \label{lewy koniec}
					\State $\pocz[i+1] \gets \pocz[i] + 1$ \Comment{Shorten the factor definition} \label{popping left}
					\State $\pocz[i] \gets \myFalse$	\Comment{Make $w[i]$ a free letter}				
				\Else \Comment{Good factor}
					\State{$j\gets \pocz[i]$} \Comment{Factor's definition begins at $j$} 
					\Repeat \Comment{Copy the pairing from the factor definition} \label{kopiowanie parowania}
						\State $\pair[i] \gets \pair[j]$ \label{copy pairing}
						\State $i \gets i+1$, $j \gets j+1$
					\Until{$\kon[i-1]$}
					\While{$\pair[i-1] \neq \prawy$} \Comment{Looking for a new end of the factor} \label{prawy koniec}
						\State $i \gets i - 1$ \label{popping right}
						\State $\kon[i-1] = \myTrue$, $\kon[i] = \myFalse$  \Comment{Shortening of the factor}	
						\State $\pair[i] = \none$ \Comment{Clear the pairing}
					\EndWhile
				\EndIf				
			\EndIf \label{end of factor}
			\If{$\Not$ $\pocz[i]$} \Comment{$w[i]$ a free letter} \label{pairing trivial factors}
				\If{$\pair[i-1] = \none$} \Comment{If previous letter is not paired}
					\State $\pair[i-1] \gets \lewy$, $\pair[i] \gets \prawy$ \Comment{Pair them} \label{pairing trivial factors 2}
				\Else
					\State $\pair[i] \gets \none$ \Comment{Leave the letter unpaired}
				\EndIf
				\State $i \gets i+1$
			\EndIf
		\EndWhile
	\end{algorithmic}
\end{algorithm}

\subsection*{Creation of pairing}
We read $w$ from left to right, suppose that we are at position $i$.
If we read $w[i]$ that is a free letter then we check, whether the previous letter is not paired.
If so, then we pair them. Otherwise we continue to the next position.

If $i$ is a first letter of a factor, we first check whether the length of this factor is one;
if so, we change $w[i]$ into a free letter.
If the factor has definition only one position to the left (i.e.\ at $i-1$) then we split the factor:
we make $w[i]$ a free letter and set $w[i+1]$ as a first letter of a factor with a~definition starting at $i-1$.
Otherwise we check whether $w[\pocz[i]]$ is indeed the first letter of a pair.
If not (i.e.\ it is a second letter of a pair or an unpaired letter) then we split the factor:
we make $w[i]$ a free letter and $w[i+1]$ the beginning of a factor with a definition beginning at $\pocz[i]+1$
(note that this factor may have length $1$);
we view the factor beginning at $w[i+1]$ as a modified factor that used to begin at $w[i]$.
If for any reason we turned $w[i]$ into a free letter, we re-read this letter, treating it accordingly.
If $w[\pocz[i]]$ is a first letter of a pair, we copy the pairing from the whole factor's definition to the factor starting at $i$.

When this is done we need to make sure that the factor indeed ends with a pair, i.e.\ that \Pref{2} holds:
if the last letter of a factor, say $w[i']$, is not the second in the pair.
To this end we split the factor: we make $w[i']$ a free letter, clear $i'$'s pairing, decrease $i'$ by $1$ and set the flag for $w[i']$
(making it the end of the new factor).
We iterate it until the $w[i']$ is indeed a second letter of a factor.
This is all formalised in Algorithm~\ref{alg: pairing}.

\begin{algorithm}[h]
	\caption{\algreplace}
	\label{alg: replace}
	\begin{algorithmic}[1]
		\State $i \gets i' \gets 1$ \Comment{$i'$ is the position corresponding to $i$ in the new word}
		\While{$i \leq |w|$}
			\If{$\pocz[i]$} \Comment{$w[i]$ is the first element of a factor}
				\State $\pocz[i'] \gets j' \gets \pos[\pocz[i]]$ \Comment{Factor in new word begins at the position corresponding to the beginning of the current factor}
				\State $\pocz[i] \gets \myFalse$ \Comment{Clearing obsolete information}
				\Repeat
					\State $\pos[i] \gets i'$ \Comment{Position corresponding to $i$}
					\State $w[i'] \gets w[j']$ \Comment{Copy the letter according to new factorisation}
					\State $i' \gets i'+1$, $j' \gets j'+1$
					\If{$\pair[i] = \lewy$} 
						\State $i \gets i+2$ \Comment{We move left by the whole pair}
					\Else
						\State $i \gets i+1$ \Comment{We move left by the unpaired letter}
					\EndIf
				\Until{$\kon[i-1]$} \Comment{We processed the whole factor}
				\State $\kon[i'-1] \gets \myTrue$ \Comment{End in the new word}
				\State $\kon[i-1] \gets \myFalse$ \Comment{Clearing obsolete information}
			\EndIf
			\If{$\Not$ $\pocz[i]$} \Comment{$w[i]$ a free letter}
				\State $\pos[i] \gets i'$
				\If{$\pair[i] = \none$} 
					\State $w[i'] \gets w[i]$ \Comment{We copy the unpaired letter}
					\State $i \gets i+1$, $i' \gets i'+1$ \Comment{We move by this letter to the right}
				\Else
					\State $w[i'] \gets $ fresh letter \Comment{Paired free letters are replaced by a fresh letter}
					\State $i \gets i+2$, $i' \gets i'+1$ \Comment{We move to the right by the whole pair}
				\EndIf
			\EndIf
		\EndWhile
	\end{algorithmic}
\end{algorithm}

\subsubsection*{Using the pairing}
When the pairing is done, we read the word $w$ again (from left to right) and replace the pairs by letters.
We keep two indices: $i$, which is the pointer in the current word (pointing at the first unread letter)
and $i'$, which is a pointer in the new word, pointing at the first free position.
Additionally, when reading $i$ we store (in $\pos[i]$) the position of the corresponding letter in the new word, which is always $i'$.

If $w[i]$ is a first letter in a pair and this pair consists of two free letters,
in the new word we add a fresh letter and move two letters to the right in $w$ (as well as one position in the new word).
If $w[i]$ is unpaired and a free letter then we simply copy this letter to the new word, increasing both $i$ and $i'$ by $1$.
If $w[i]$ is first letter of a factor (and so also a first letter of a pair by \Pref{2}),
we copy the corresponding fragment of the new word (the first position is given by $\pos[\pocz[i]]$),
moving $i$ and $i'$ in parallel: $i'$ is always incremented by $1$,
while $i$ is moved by $2$ when it reads a first letter of a pair and by $1$ when it reads a free letter.
Also, we store the new beginning and end of the factor in the new word:
for a factor beginning at $i$ and ending at $i'$ 
we set $\pocz[\pos[i]] = \pos[\pocz[i]]$ and 
$\kon[\pos[i'-1]] = \myTrue$ (note that $i'-1$ and $i'$ are paired).
Details are given in Algorithm~\ref{alg: replace}.

\begin{algorithm}[h]
	\caption{\algmain}
	\label{alg: main}
	\begin{algorithmic}[1]
		\State compute LZ77 factorisation of $w$
		\While{$|w| > 1$}
			\State compute a pairing of $w$ using \algpairing
			\State replace the pairs using \algreplace 
		\EndWhile
		\State output the constructed grammar
	\end{algorithmic}
\end{algorithm}

\subsection*{Algorithm} \algmain{} first computes the LZ77 factorisation and then iteratively applies \algpairing{} and \algreplace,
until a one-word letter is obtained.

\section{Analysis}

We begin the analysis with showing that indeed \algpairing{} produces the desired pairing.

\begin{lemma}
\label{lem: main}
\algpairing{} runs in linear time. It creates a proper factorisation and returns a pairing that satisfies \Pref{1}--\Pref{3}
(for this new factorisation).
When the current factorisation for the input word for \algpairing{} has $m$ factors
then \algpairing{} creates at most $6m$ new free letters and the returned pairing has at most $m$ factors.
\end{lemma}
\begin{proof}
For the running time analysis note that a single letter can be considered at most twice:
once as a part of a factor and once as a free letter.

We show the second claim of the lemma by induction:
at all time the stored factorisation is proper, furthermore, 
when we processed $w[1\twodots i-1]$ (i.e.\ we are at position $i$,
note that we can go back in which case position gets unprocessed)
then we have a partial pairing on $w[1\twodots i-1]$,
which differs from the pairing only in the fact that the position $i-1$ may be assigned as first in the pair
and $i$ is not yet paired.
This partial pairing satisfies \Pref{1}--\Pref{3} restricted to $w[1 \twodots i-1]$.

\subsubsection*{Factorisation}
We first show that after considering $i$ the modified factorisation is proper.

If in line~\ref{ensure non-trivial} we have $\pocz[i] = \kon[i]$ then $w[i]$ is a one-letter factor and so after replacing it with a free letter
the factorisation stays proper.
Observe now that the verification in line~\ref{ensure non-trivial} ensures that in each other case considered
in lines \ref{blok}--\ref{end of factor} we deal with factors of length at least $2$.

The modifications of the factorisation in line~\ref{blok pocz} results in a proper factorisation:
the change is applied only when $\pocz[i] = i-1$, in which case $w[i\twodots i + |f| - 1] = w[i - 1 \twodots i + |f| - 2]$,
which implies that $f = a^{|f|}$, where $a = w[i]$.
Since $|f| \geq 2$ (by case assumption), in such a case $w[i+1\twodots i + |f| - 1] = w[i-1\twodots i + |f| - 3]$ so we can split the factor $f = w[i\twodots i + |f| - 1]$
to $w[i]$ and a factor $w[i+1\twodots i + |f| - 1]$ which is defined as $w[i - 1 \twodots i + |f| - 3]$
($f$ had at least two letters, so after the modification it has at least $1$ letter).

In line~\ref{popping left} we shorten the factor by one letter (and create a free letter),
as the factor had at least two letters, so the factorisation remains proper.

Concerning the symmetric shortening in line~\ref{popping right}, it
leaves a proper factorisation (as in case of line~\ref{popping left}),
as long as we do not move $i$ before the beginning of the factor.
However, observe that when we reach line~\ref{copy pairing} this means that the factor beginning at $i'$ has length at least $2$,
$\pocz[i'] < i'-1$ and $\pair[\pocz[i']] = \lewy$.
Thus $\pocz[i']+1 < i'$ and so by induction assumption we already made a pairing for it. Since $\pocz[i']$ is assigned $\lewy$,
$\pocz[i']+1$ is assigned $\prawy$.
So $i'+1$ is assigned $\prawy$ as well. Since the end of the factor is at position $i \geq i'+1$,
in our search for element marked with \prawy{} at positions $i$, $i-1$, \ldots we cannot move to the left more than to $i'+1$.
Thus the factor remains (and has at least $2$ letters).

\subsubsection*{Pairing}
We show that indeed we have a partial pairing.
Firstly, if $i$ is decreased, then as a result we get a partial pairing:
the only nontrivial case is when $i-1$ and $i$ were paired then $i-1$ is assigned as the first element in the pair but it has
no corresponding element, which is allowed in the partial pairing.
If $i$ is increased then we need to make sure if $i-1$ is assigned as a first element in a pair then $i$ will be assigned as the second one
(or the pairing is cleared).
Note that $i-1$ can be assigned in this way only when it is part of the factor, i.e.\ it gets the same status as some $j$.
If $i$ is also part of the same factor, then it is assigned the status of $j+1$, which by inductive assumption is paired with $j$,
so is the second element in the pair.
In the remaining case, if $i-1$ was the last element of the factor then in loop in line~\ref{prawy koniec} we decrease $i$ and so unprocess $i-1$
(in particular, we clear its pairing).

For \Pref{2} observe that for the first two letters it is explicitly verified in line~\ref{lewy koniec}.
Similarly, for the second part of \Pref{2}: we shorten the last factor in line~\ref{prawy koniec}
(ending at $i$) until $\pair[i] = \prawy$.
We already shown that pairing is defined for $w[1\ldots i]$ and when $i$ is assigned $\prawy$ then $i-1$ is assigned $\lewy$, as claimed.

Condition \Pref{3} is explicitly enforced in loop in line~\ref{kopiowanie parowania},
in which we copy the pairing from the definition of the factor.

Suppose that \Pref{1} does not hold for $i-1,i$, i.e.\ they are both unpaired after processing $i$.
It cannot be that they are both within the same factor,
as then the corresponding $w[j-1]$ and $w[j]$ in the definition of the factor are also unpaired, by \Pref{3},
which contradicts the induction assumption.
Similarly, it cannot be that one of them is in a factor and the other outside this factor,
as by \Pref{2} (which holds for $w[1\twodots i]$) a factor begins and ends with two paired letters.
So they are both free letters.
But then we needed to pass line~\ref{pairing trivial factors} for $i$ and both $w[i-1]$ and $w[i]$ were free at that time,
which means that they should have been paired at that point, contradiction.

To see the third claim of the lemma, i.e.\ the bound on the number of new free letters,
fix a factor $f$ that begins at position $i$.
When it is modified, we identify the obtained factor with $f$
(which in particular shows that the number of factors does not increase).
We show that it creates at most $6$ new free letters in this phase.

If at any point $\pocz[i]$ and $\kon[i]$, i.e.\ the factor has only one letter,
then it is replaced with a free letter and afterwards cannot introduce any free letters (as $f$ is no longer there).
Hence at most one free letter is introduced by $f$ due to condition in line~\ref{ensure non-trivial}.

If $\pocz[i] = i-1$ then it creates one free letter inside condition in line~\ref{blok}.
It cannot introduce another free letter in this way (in this phase),
as afterwards $\pocz[i+1] = i-1$ and there is no way to decrease this distance (in this phase).

We show that condition in line~\ref{lewy koniec}, i.e.\ that the first letter of the factor definition is not the first in the pair,
holds at most twice for a fixed factor $f$ in a phase.
Since we set $j = \pocz[i]$ and increase both $i$ and $j$ by $1$ until $\pair[j] = \lewy$,
this can be viewed as searching for the smallest position $j' \geq j$ that is first in a pair
and we claim that $j' \leq j+2$.
On the high-level, this should hold because \Pref{1} holds for $w[1 \twodots i-1]$,
and so among three consecutive letters there is at least one that is the first element in the pair.
However, the situation is a bit more complicated, as some pairing may change during the search.%
\footnote{In particular, this could fail if we allowed that $\pocz[i] = i-1$, so some care is needed.}

Consider first the main case, in which $j \leq i-3$. Then the elements at position $j$, $j+1$, $j+2 \leq i-1$
have already assigned pairings and so at least one of them is assigned as a first element of some pair.
The only way to change the pairing from $\lewy$ to some other is in loop in line~\ref{prawy koniec}.
However, we can go to this loop only after condition in line~\ref{lewy koniec} fails,
which implies that it holds at most twice (i.e.\ for at most two other among $j$, $j+1$, $j+2$).

As the case in which $\pocz[i] = i-1$ is excluded by the case assumption of the algorithm, the remaining case is $j = i-2$.
As in the previous argument, we consider the letters whose pairing are known, i.e.\ $w[i-2]$ and $w[i-1]$.
If any of them is a first letter in a pair, we are done (as in the previous case).
As \Pref{1} holds for them, the only remaining possibility is that $w[i-2]$ is a second letter in a pair and $w[i-1]$ is unpaired.
Then when we consider $i$ in line~\ref{lewy koniec} it is made free.
When we consider $w[i]$ in line~\ref{pairing trivial factors} (re-reading it as a free letter), it is paired with $i-1$.
Hence when we read $i+1$ (the new first letter in the factor) its definition ($i-1$) is a first letter in a pair, as claimed.

Similar analysis can be applied to the last letter of a factor.
So, as claimed, one factor can introduce at most $6$ free letters in a single phase.

Finally, it is left to show that when we processed the whole $w$ then we have a proper factorisation and a pairing satisfying \Pref{1}--\Pref{3}.
From the inductive proof it follows that the kept factorisation is proper and the partial pairing satisfies \Pref{1}--\Pref{3} for the whole word.
So it is enough to show that this partial pairing is a pairing, i.e.\ that the last letter of $w$ is not assigned as a first element of a pair.
Consider, whether it is in a factor or a free letter.
If it is in a factor then clearly it is the last element of the factor and so by \Pref{2} it is the second element in a pair.
If it is a free letter observe that we only pair free letters in line~\ref{pairing trivial factors 2},
which means that it is paired with the letter on the next position, contradiction.
\qedhere
\end{proof}

Now, we show that when we have a pairing satisfying \Pref{1}--\Pref{3}
(so in particular the one provided by \algpairing{} is fine, but it can be any other pairing satisfying \Pref{1}--\Pref{3})
then \algreplace{} creates a word $w'$ out of $w$ together with a factorisation.

\begin{lemma}
\label{lem: pairing and trivial}
When a pairing satisfies \Pref{1}--\Pref{3} then \algreplace{} runs in linear time and returns a word $w'$ together with a factorisation;
$|w'| \leq \frac{2|w|+1}{3}$ and the factorisation of $w'$ has the same number of factors as the factorisation of $w$.
If $p$ fresh letters were introduced then $w'$ has $p$ less free letters than $w$.
\end{lemma}
\begin{proof}
The running time is obvious as we make one scan through $w$.

Firstly, we show that when we erase the information about beginnings and ends of factors of $w$
we do not erase the newly created information for $w'$.
To this end it is enough to show that $i > i'$ in such situation.
Whenever $i'$ is incremented, $i$ is incremented by at least the same amount, so it is enough to show that $i > i'$
when $i$ is the first letter of the first factor,
in other words, there is at least one pair before the first factor.
By \Pref{1} there is a pair within first three positions of the factor.
If the pair is at positions $1$, $2$ then by \Pref{2} the factor begins at position $3$ or later and we are done.
If the pair is at positions $2$, $3$ then by \Pref{2} the factor begins at position $4$ or later or at position $2$;
however, the latter case implies that the factor definition is $1$ to the left of the factor, which is excluded by \algpairing.

Concerning the size of the produced word, by \Pref{1} each unpaired letter (perhaps except the last letter of $w$) is followed by a pair.
Thus, at least $\frac{1}{3}(|w|-1)$ letters are removed from $w$, which yields the claim.

Concerning the factorisation of $w'$,
observe that by an easy induction it can be shown that for each $i$ the $\pos[i]$ is 
\begin{itemize}
	\item undefined, when $i$ is second in a pair \textit{or}
	\item is the position of the corresponding letter in $w'$.
\end{itemize}

Now, consider any factor $f$ in $w$ with a definition $w[j \ldots j + |f| - 1]$.
By \Pref{2} both the first and the last two letters of $f$ are paired
and by \Pref{3} pairing of $f$ is the same as the pairing of its definition.
So it is enough to copy the letters in $w'$ corresponding to $w[j \ldots j + |f| - 1]$, i.e.\ beginning with $\pos[j]$,
which is what the algorithm does.
When we consider a free letter, if it is unpaired, it should be copied (as it is not replaced),
and when it is paired, the pair can be replaced with a fresh letter; in both cases the corresponding letter in the new word should be free.
And the algorithm does that.

Concerning the number of fresh letters introduced, suppose that $ab$ is replaced with $c$.
If $ab$ is within some factor $f$ then we use for the replacement the same letter as we use in the factor definition
and so no new fresh letter is introduced.
If both this $a$ and $b$ are free letters then each such a pair contributes one fresh letter.
And those two free letters are replaced with one free letter, hence the number of free letters decreases by $1$.
The last possibility is that one letter from $ab$ comes from a factor and the other from outside this factor,
but this contradicts \Pref{2} that a factor begins and ends with a pair.
\qedhere
\end{proof}

Using the two lemmata we can give the proof.

\begin{theorem}
\algmain{} runs in linear time and returns an SLP of size $\Ocomp(\ell + \ell \log(N/\ell))$.
Thus its approximation ration is $\Ocomp(1 + \log(N/g))$, where $g$ is the size of the optimal grammar.
\end{theorem}
\begin{proof}
Due to Lemma~\ref{lem: pairing and trivial} each introduction of a fresh letter reduces the number of free letters by $1$.
Thus to bound the number of different introduced letters it is enough to estimate the number of created free letters.
In the initial LZ77 factorisation there are at most $\ell$ of them.
For the free letters created during the \algpairing{} let us fix a factor $f$ of the original factorisation
and estimate how many free letters it created.
Due to \Pref{1} the length of $f$ drops by a constant fraction in each phase and so it will take part in $\Ocomp(\log |f|)$ phases.
In each phase it can introduce at most $6$ free letters, by Lemma~\ref{lem: main}.
So $\Ocomp(\sum_{i=1}^\ell \log |f_i|)$ free letters were introduced to the word during all phases.
Consider $\sum_{i=1}^\ell \log |f_i|$ under the constraint $\sum_{i=1}^\ell |f_i| \leq N$.
As function $h(x) = \log(x)$ is concave, we conclude that this is maximised for all $f_i$ being equal to $N/\ell$.
Hence the number of nonterminals in the grammar introduced in this way is $\Ocomp(\ell \log(N /\ell))$.
Adding the $\ell$ for the free letters in the LZ77 factorisations yields the claim.

Concerning the running time, the creation of the LZ77 factorisation takes linear time~\cite{CrochemoreLZ77,PuglisiLZ77,CrochemoreLZ772,BannaiLZ77,KarkkainenLZ77,OhlebuschLZ77}.
In each phase the pairing and replacement of pairs takes linear time in the length of the current word.
Thanks to \Pref{1} the length of such a word is reduced by a constant fraction in each phase,
hence the total running time is linear.
\qedhere
\end{proof}

%%%\bibliographystyle{plain}
%%%\bibliography{references}
\end{document}